\pdfoutput=1
\newif\ifFull
\Fulltrue
\ifFull
\documentclass[11pt]{article}
\usepackage{fullpage}
\usepackage{color}
\newenvironment{proof}{\noindent{\bf Proof:}}{\hspace*{\fill}\rule{6pt}{6pt}\bigskip}
\else
\documentclass[graphicx]{acm_proc_article-sp}
\conferenceinfo{WPES'11,} {October 17, 2001. Chicago, Illinois, USA.}
\CopyrightYear{2011}
\crdata{978-xxx}
\fi

\usepackage{algorithmic}
\usepackage{url}
\usepackage{graphicx}
\makeatletter
\def\@begintheorem#1#2{\sl \trivlist \item[\hskip \labelsep{\bf #1\ #2:}]}
\def\@opargbegintheorem#1#2#3{\sl \trivlist
      \item[\hskip \labelsep{\bf #1\ #2\ #3:}]}
\makeatother

\newtheorem{theorem}{Theorem}

\begin{document}

\title{Privacy-Enhanced Methods for\\ Comparing Compressed DNA Sequences}

\ifFull
\author{
David Eppstein \\[5pt]
{Dept.~of Computer Science} \\ 
{University of California, Irvine} \\
{Irvine, California 92697-3435 USA} \\[3pt]
{eppstein(at)ics.uci.edu}
\and
Michael T. Goodrich \\[5pt]
{Dept.~of Computer Science} \\ 
{University of California, Irvine} \\
{Irvine, California 92697-3435 USA} \\[3pt]
{goodrich(at)ics.uci.edu}
\and
Pierre Baldi \\[5pt]
{Dept.~of Computer Science} \\ 
{University of California, Irvine} \\
{Irvine, California 92697-3435 USA} \\[3pt]
{pfbaldi(at)ics.uci.edu}
}
\else
\numberofauthors{3}
\author{
\alignauthor David Eppstein \\[5pt]
\affaddr{Dept.~of Computer Science} \\ 
\affaddr{University of California, Irvine} \\
\affaddr{Irvine, California 92697-3435 USA} \\[3pt]
\email{eppstein(at)ics.uci.edu}
\alignauthor Michael T. Goodrich \\[5pt]
\affaddr{Dept.~of Computer Science} \\ 
\affaddr{University of California, Irvine} \\
\affaddr{Irvine, California 92697-3435 USA} \\[3pt]
\email{goodrich(at)ics.uci.edu}
\alignauthor Pierre Baldi \\[5pt]
\affaddr{Dept.~of Computer Science} \\ 
\affaddr{University of California, Irvine} \\
\affaddr{Irvine, California 92697-3435 USA} \\[3pt]
\email{pfbaldi(at)ics.uci.edu}
}
\fi

\date{}

\maketitle 

\begin{abstract}
In this paper, we study methods for improving the efficiency and privacy
of compressed
DNA sequence comparison computations, under various querying scenarios.
For instance, one scenario involves a querier, Bob, who wants to
test if his DNA string, $Q$, is close to a DNA string, $Y$, owned by a data
owner, Alice, but Bob does not want to reveal $Q$ to Alice and Alice
is willing to reveal $Y$ to Bob \emph{only if} it is close to $Q$.
We describe a privacy-enhanced method for comparing two compressed
DNA sequences, which can be used to achieve the goals of such a scenario.
Our method involves a reduction to set differencing, and
we describe a privacy-enhanced protocol for set differencing that
achieves absolute privacy for Bob (in the information theoretic
sense), and a quantifiable degree of privacy 
protection for Alice.
One of the important features of our protocols, which makes them
ideally suited to privacy-enhanced DNA sequence comparison problems,
is that the communication complexity of our solutions is proportional
to a threshold that bounds the cardinality of the set differences that
are of interest, rather than the cardinality of the sets involved (which
correlates to the length of the DNA sequences).
Moreover, in our protocols, the querier, Bob, can easily compute the 
set difference only if its cardinality is close to or below a
specified threshold.
\end{abstract}

\ifFull\else
%\category{F.2.0}{Analysis of Algorithms and Problem Complexity}{General}


\fi

\pagestyle{plain}
\def\thepage{\arabic{page}}

\section{Introduction}
\medskip\noindent
It is hard to imagine the meaning of privacy in a world where someone
can download a digital copy of your DNA.
Given this data and your identity, such a person
could determine who you are, your ethnic heritage, and what 
diseases you are likely to inherit. 
There are already significant
concerns that employers or insurers will use genetic information to screen 
those at high risk for various diseases, and
several governments
have already created laws dealing with DNA data access.
Thus, there is a need for technologies that can safeguard the privacy 
and security of genomic data.
Unfortunately, existing methods for privacy-preserving data querying,
which depend on expensive cryptographic computations,
don't scale to genomic data sizes.

Indeed, the size of genomic data is not just a concern for security
primitives---the storage and handling of genomic data is already
presenting a challenge even without considering the possible
addition of cryptographic layers.
Namely, genomes are often stored as simple text 
files, but this approach is difficult to scale in many ways. 
The storage of the diploid genomes of all currently living humans
using this simple approach would be
on the order of $36 \times 10^{18}$ bytes, or 
36 exabytes, not counting likely backups and multiple copies.
And even with the progress that can be expected with increased
capacities for storage and networking in the coming years, 
it is likely that security and privacy issues will require 
additional layers of protection around genomic data.
Thus, it is quite likely that genomic data will need to be processed
and stored in compressed form.

Many companies, physicians, and law-enforcement agencies
have legitimate
reasons for querying genomic data. 
The challenge from an algorithmic
standpoint is to be able to support such queries using low storage,
bandwidth, and query times.
Moreover, as just mentioned,
it is our view
that it is quite likely that genomic data will have to be
compressed, given their uncompressed sizes.
Thus, in order to protect the privacy of genomic data,
cryptographic techniques will need to be employed on compressed
genomic data to support
privacy-preserving data querying. 
That is, cryptographic techniques should
allow for queries to be performed on compressed data
in a way that answers the specific
question---such as quantifying 
sequence matches or alignments---but does not
reveal any other information about the data, such as race or disease risk of
the individual whose DNA is being queried.
Putting these constraints together, there is a real need for
trustworthy computing
techniques that can work in conjunction with genomic compression
technologies to answer queries and perform analyses
in a way that preserves the privacy of the data and the proprietary
information contained in the queries themselves.

\subsection{Prior Related Work}
\medskip\noindent
Privacy issues for data querying have been investigated before.
Namely, the data structures presented in~\cite{mrk-zks-03, ors-ecpgqcd-04}
provide the option of a privacy-preserving verification of the answer to a
query, that is, the proofs of data consistency are carried out in a private
way, where no information other than the queried data is inferred by the
prover.  However, these privacy-preserving techniques involve
computationally expensive operations, which would be inappropriate
for genomic-scale computations.

Several researchers have 
also explored privacy-preserving data querying
methods that can be applied, in principle, to
genomic sequences (e.g., see~\cite{akd-spsc-03,da-smc-01}).
That is, cryptographic techniques 
can be used to allow for queries 
to be performed in a way that answers the specific
question---such as a score rating the quality of a query
for DNA matching or sequence alignment---but does not
reveal any other information about the data, such as race or disease risk of
the individual whose DNA is being queried.
Atallah {\it et al.}~\cite{akd-spsc-03} and
Atallah and Li~\cite{al-sosc-05}
studied
privacy-preserving protocols for edit-distance sequence comparisons,
Troncoso-Pastoriza {\it et al.}~\cite{tkc-pperd-07}
described a privacy-preserving protocol for regular-expression
searching.
Jha {\it et al.}~\cite{jks-tppgc-08} give
privacy-preserving protocols for computing edit distance and
Smith-Waterman similarity scores between two sequences,
improving the privacy-preserving 
algorithm of Szajda {\it et al.}~\cite{spol-tpdp-06}.
These previous schemes are not defined for compressed DNA sequences,
however, and they have communication costs that are as high as the
strings themselves, whereas we are interested in solutions that have
communication costs proportional to the size of the differences
between DNA sequences.

Aligned matching results between two compressed genomic strings can 
also be done in a secure way using 
privacy-preserving set-intersection protocols, and
several groups of researchers have developed such 
protocols
(e.g.,
see~\cite{ae-nepps-07,FNP04,vc-ssica-05,ss-ppsip-07,ss-ppsib-08})
or SMC methods for
computing dot products (e.g., see~\cite{dfknt-uscfm-06,gmw-hpamg-87,y-psc-82}),
which can also be used for determining set intersections.
Freedman {\it et al.}~\cite{FNP04}
give a linear-size protocol for performing a privacy-preserving 
set-intersection computation,
and several others have given privacy-preserving solutions
with linear or slightly super-linear communication costs, including
Dachman-Soled {\it et al.}~\cite{dmry-erpsi-09} and
De~Cristofaro and Tsudik~\cite{dt-ppsip-10}.
Instead of computing the contents of a set intersection, 
De~Cristofaro {\it et al.}~\cite{cryptoeprint:2011:141} give a
protocol having communication complexity proportional to the set
sizes for simply determining the cardinality of a set intersection,
and Vaidya and Clifton~\cite{vc-ssica-05} similarly give such a
privacy-preserving protocol with linear communication complexity as well.
Kissner and Song~\cite{ks-ppso-05} 
and Hazay and Nissim~\cite{hn-esopm-10} 
extend privacy-preserving set-intersection 
protocols
to perform set union as well, also with linear communication
complexity.
Ateniese {\it et al.}~\cite{adt-ismsh-11} describe a scheme that uses
communication that is linear in one of the two sets being compared,
while keeping the other size secret.

All of these families of privacy-preserving
protocols are unfortunately inefficient for privacy-enhanced
comparison of even compressed DNA sequences, however, for two reasons:
\begin{enumerate}
\item
We desire solutions that identify the differences between two sets,
not their intersection.
\item
We desire solutions whose communication complexity is proportional to
the size of the set difference, not the sizes of the sets themselves
(which will, in the DNA application, be much larger than the set
difference).
\end{enumerate}

\subsection{Our Results}
\medskip\noindent
In this paper, we describe a new data structure, called the 
\emph{Privacy-Enhanced Invertible Bloom Filter} (PIBF), and we show how it
is suited to perform efficient privacy-enhanced comparisons of
compressed DNA strings.
In particular, we consider scenarios in which a data querier,
Bob, wishes to learn the difference between a genomic sequence of
interest and the DNA strings owned by a data owner, Alice.
In each scenario we consider, Bob does not want Alice to learn the
contents of his query string and Alice does not want to reveal any of
her DNA strings to Bob unless it is close in edit distance to
Bob's string.
We also consider a scenario where
the result of a query should only be revealed
to a trusted third party, Charles.
In addition, we can generalize our solutions so that they can apply
to queries that are restricted to a range, $R$, of the genome.
We show that the PIBF data structure can be used to
answer the queries in 
each of these scenarios in a way that protects the privacy interests
of each of the parties in quantifiable ways.
Our solutions provide absolute privacy protection to Bob
(in the information theoretic sense, for our most likely of scenarios)
and involve 
a novel quantified analysis of the privacy protection provided by
the Invertible Bloom Filter~\cite{eg-sirtd-11,egs-sos-11} data structure
for Alice.

Our solutions can, in fact, apply to scenarios that go beyond the application
to genomic data, and can be used to 
provide privacy-enhanced methods for
performing general set difference operations, which may be of
independent interest.
We focus here on the application to compressed DNA sequences, however,
since it is arguably the best motivation for a privacy-enhanced
set-differencing protocol whose communication complexity is
proportional to the cardinality of the set difference rather than the
sets being compared.
Also, the computational requirements of our solutions are very low,
even with respect to constant factors, which also makes the
application to genomic data especially relevant.

\section{Genomic Data Compression}
\medskip\noindent
Since our methods are to be applied to compressed genomic sequences,
we discuss here an effective representation for such compressed
sequences~\cite{baldicompression07,baldidcc08,brandon2009data}.
Our methods for privacy-preserving comparisons work with other
compression schemes, as well, but to provide a concrete
example of genomic data compression, let us review
the approach of~\cite{baldicompression07,baldidcc08,brandon2009data}.
The essential property to note during this discussion is that the
scheme we describe
here allows us to reduce sequence comparison to a symmetric
difference operation on sets.

Such data structures allow the compression of genome sequences 
while facilitating certain classes of sequence queries.
More importantly, at least for the scope of this paper, we assume
these data structures and compressed representations
facilitate
efficient protocols for privacy-preserving genomic data querying.

\subsection{Compressing with a Reference String}
\medskip\noindent
As noted above,
in the case of storing multiple genomes from the same species, 
and in particular for humans,
whose genomes contain at least 3 billion base pairs,
the flat text file approach is clearly wasteful.
A simple, but general, approach is to store a reference sequence, 
$\cal R$,
and then for each other sequence, encode only its differences with 
respect to $\cal R$ in a canonical way.
To see what we mean by a canonical encoding of a difference with
respect to a reference string, consider the sequences
\begin{quotation}
AACGACTAGTAATTTG 
\end{quotation}
and  
\begin{quotation}
CACGTCTAGTAATGTG,
\end{quotation}
which are identical, except for substitutions
in positions 1 (A$\to$C), 5 (A$\to$T), and 14 (T$\to$G). Each 
such single nucleotide polymorphisms (SNP)
can be encoded by a pair $(i,X)$, where $i$ is an integer encoding the position and $X$ represents the value of the substitution.
Thus, given the first sequence 
as a reference, the second one can be encoded by the string 
``1C5T14G'',
or the set 
\begin{quotation}
``$\{$(1,C), (5,T), (14,G)$\}$.''
\end{quotation}
Note that with this data representation, the questions 
``Is this sequence different from the reference sequence at position $i$?''
and, if so, ``How?'' are easy to answer, given the reference string,
$\cal R$. 

Other events, such as deletions and insertions can also 
be accommodated in this scheme. 
For a deletion, we can use a pair of integers,
$(i,l)$, where $i$ denotes the index in $\cal R$ 
where the deletion occurs and $l$ represents the length of the deletion.
Likewise, for an insertion of $l$ characters, we can use the encoding
$i, X_1 \ldots X_l$ to denote the insertion of $X_1 \ldots X_l$ 
at position $i$ with respect to the reference sequence.

Thus, with this approach, we can represent any genomic sequence as a
set, of substitution, insertion, and deletion events, and we assume
that there is a canonical way of defining such a set.
In addition,
this compression scheme is just one data representations we can allow.
The essential properties of such a compression scheme is that
it provide an effective mechanism for compressing 
multiple genomic sequences taken from the same species by encoding
each genomic sequences as an indexed set of differences with a
reference string, $\cal R$.

\subsection{Some Technical Details}
\medskip\noindent
While the basic idea of genomic data compression is often easy to understand, 
precise implementations require that one addresses a number of important 
technical issues, for the sake of defining a canonical encoding that
can allow us to reduce genomic sequence comparison to set
differencing.
For instance, using the above approach as a running example, we
need to use integers that are associated with positions or 
lengths of events, 
and letters describing these events.
Our set comparison methods depend on indices that are 
absolute addresses, but a compression scheme might not use such
addresses.

For instance,
one can use local coordinates or relative addressing, i.e., intervals, 
instead of absolute addressing, and get improved data compression. 
With relative addresses, applied to the above example, the string 
``1C5T14G'' becomes ``0C4T9G'', with a similar change to the
associated set representation.
It is natural to expect such an encoding, since
the integers to be encoded may be
considerably smaller than when one uses absolute addresses.
For our comparison operations to be effective, we need to be using
absolute coordinates.
Thus,
the relatively modest price to pay is that if a genomic sequences is
encoded using relative addresses, then these addresses must be 
changed to recover absolute coordinates prior to applying our
privacy-enhanced comparison methods.

A second observation is that if the positions at which variations 
occur in the population are fixed and
form a relatively small subset of all possible positions, 
then additional savings may result by focusing only on those positions.
If in the same schematic example as above, one knew that in the 
population substitutions can only occur at positions 1, 5, and 14,
one could, for instance, encode ``1C5T14G'' by ``1C2T3G'', 
at the cost of keeping a conversion table that memorizes the 
coordinate positions where the variants occur.
In this case, these integer values form a modified type of 
absolute address; hence, they do not need to be modified in order for
our privacy-preserving comparison methods to be effective.
That is, in the genome applications we consider, 
the absolute addresses should either be explicit positional indices 
or ordinal indices with respect to known mutations in the population
(e.g., see~\cite{baldicompression07,baldidcc08,brandon2009data}).

While substitutions, deletions, and insertions are sufficient both in 
theory to describe any sequence with respect to the reference sequence 
and in practice in the case of the relatively short mitochondrial genome,
more complex events that occur frequently in more complex genomes 
can also be added to the list. 
For instance, 
\begin{quotation}
``50Ins100C2(25)'' 
\end{quotation}
could denote the insertion (Ins) 
at position 50 of a subsequence of length 100 bp 
coming from chromosome 2, and starting at location 25 on that chromosome.
Clearly the binary convention for denoting the different pieces of 
information, or their groups, must be defined to avoid both any ambiguity 
as well as the use of additional symbols such as parentheses 
or other delimiters.
While statistics on human SNPs are becoming 
available~\cite{hinds05,goldstein05,HapMap07}, 
statistics on the frequencies and location of these more complex
rearrangement events in the human genome are at an earlier stage of development,which can be used for possible improved DNA sequence compression.
(For instance, see~\cite{levy07}.)
In any case,
the approaches used in this paper can also be applied to such
compression schemes, provided
they preserve the invariant that DNA comparison amounts to a
symmetric difference between compressed DNA sequences (converted to
absolute coordinates) defined by sets of differences with a reference
string, $\cal R$.

\subsection{Data Distributions}
\medskip\noindent
As noted above,
genomic strings typically have a degree of similarity
that can be exploited for compressing such schemes. Indeed, 
we have already noted that 
compression schemes can let us view a genomic string
with respect to a scheme that represents a string in terms of its
differences with a reference string, $\cal R$
(e.g., see~\cite{baldicompression07,brandon2009data}).
That is, we can start from a reference string, $\cal R$, 
which contains the most common components of a typical genomic string.
Then we define each other string, $Q$, in a canonical way
in terms of its differences with $\cal R$.
Each difference is defined by an index location, $i$, in $\cal R$ 
and an operation
to perform at that location, such as a substitution, insertion, or deletion.

In~\cite{brandon2009data},
a comparative analysis on a set of 4,000 mtDNA sequences from different
individuals, taken from GenBank, has identified
4,544 positions along the reference mtDNA 
sequence where at least one of the other sequence deviates from the 
reference sequence, which has a total length of roughly 16,500. 
In aggregate, there were
122,266 bp that deviated from the reference sequence.
Besides substitutions, the total number of insertion 
and deletion events across all the sequences was 7,175, 
the most frequent one being
1 bp insertions (4671 occurrences), followed by 2 bp deletions (901). 
In addition, by combining the computations that determined these
statistics with machine learning algorithms for racial
characteristics, researchers have identified ethnic mtDNA deviations.
Some well known variants, 
such as the ``Asian-specific 9 bp deletion'' \cite{harihara92,thomas98}, 
also occur frequently (255 occurrences).
The take-away message from this experiment was that canonical
data compression is viable using the approach outlined above.

An important observation that already 
can be derived from this study of mtDNA is that DNA strings that have
a large edit distance will also have a large number of differences
with respect to their compressions relative to the reference string.
For instance, the distribution of the raw intervals 
for the mtDNA study
is shown in Figure \ref{fig:runlength}a. Observed intervals vary from 0 to
14,998bp, the most frequent one being an interval of 73 (2,580 occurrences), followed by
688 (2419 occurrences), and followed by 6 (2,202 occurrences).
Likewise,
the plot of the logarithm of the counts versus the logarithm of the rank 
(in decreasing order of frequency) is shown in Figure
\ref{fig:runlength}b. 
Overall these distributions are not strongly structured, with the log-log plot
demonstrating in this case at best a very weak power-law structure.

\begin{figure*}[!htb]
\begin{center}
\begin{tabular}{cc}
\includegraphics[width=3.0in]{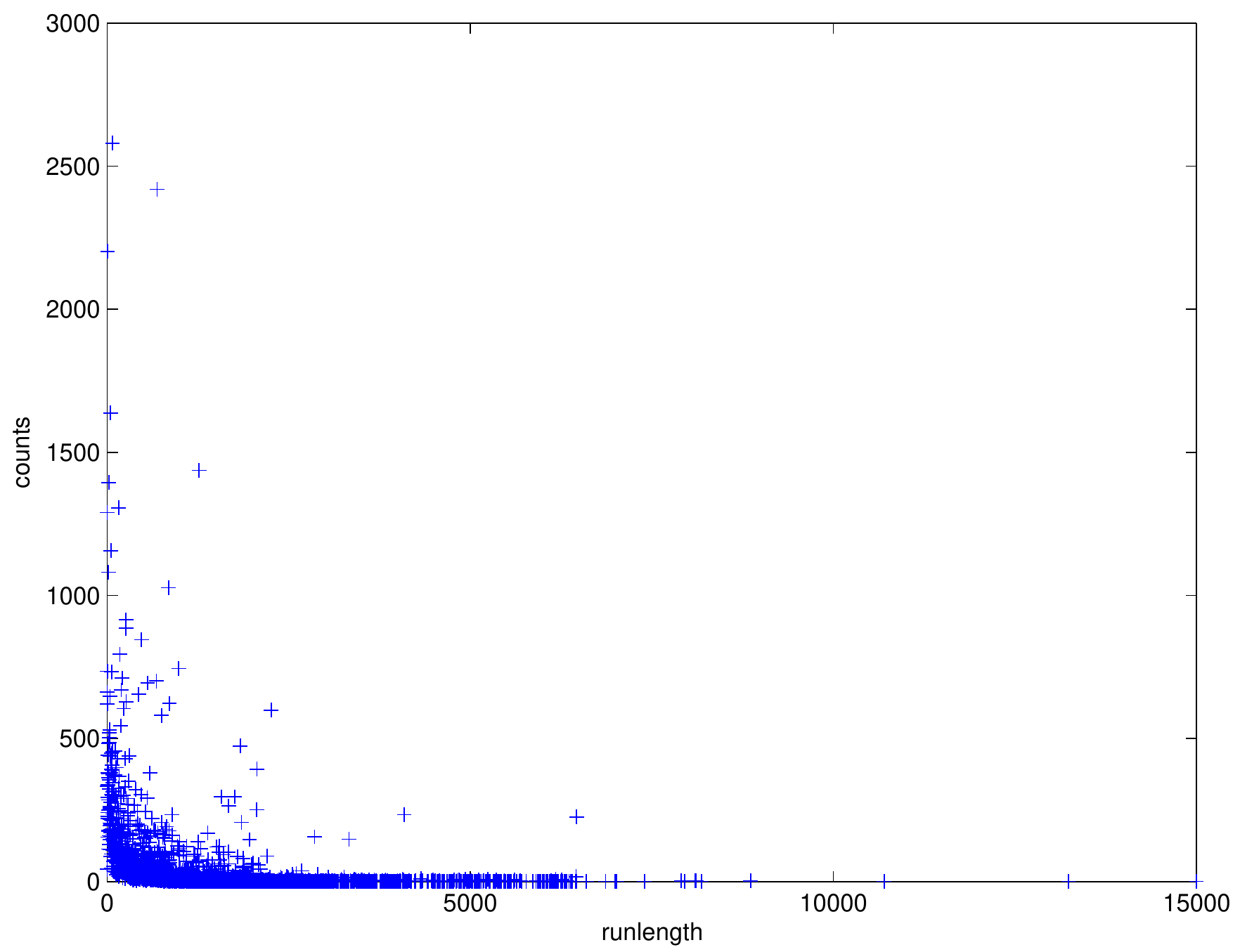}
&
\includegraphics[width=3.0in]{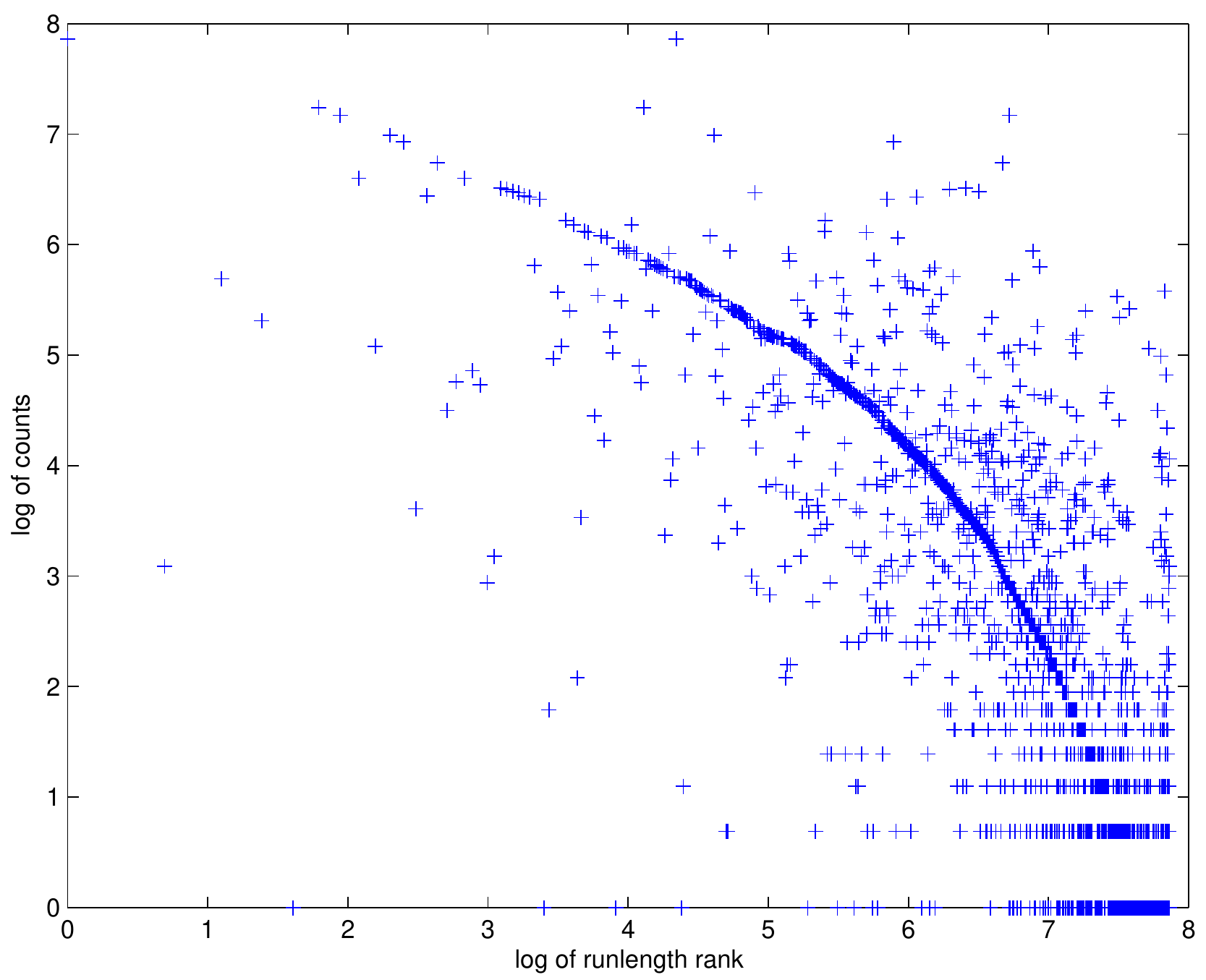}
\\
(a)
&
(b)
\end{tabular}
\end{center}
\caption{Distribution of intervals. 
(a) the $x$-axis represents the interval values and
the $y$-axis represents the corresponding counts;
(b) distribution of intervals using a log rank-log frequency plot,
with $x$-axis representing the the logarithm of the rank associated 
with decreasing interval frequencies and the
$y$-axis representing the logarithm of the corresponding counts.
}
\label{fig:runlength}
\end{figure*}

Thus, this exercise shows
that the reference sequence can be optimized 
to minimize the total number of variants.
Furthermore, the reference sequence does
not need to be a sequence from an actual individual, 
but could be designed using purely statistical and 
entropy minimization considerations.
The design of the reference sequence impacts not 
only the variants to be recorded, but also the runlengths, 
and therefore it must
also take into consideration any constraints a particular 
implementation may place on the runlengths and their encodings. 
In the case
of large genomes, for instance, it may be useful to control 
the range of possible runlengths and keep it under some reasonable value.
The only requirement that we make with respect to the comparison
algorithms we consider in this paper is that the same reference
string should be used to encode all DNA strings that will be
subjected to privacy-enhanced comparisons.
Moreover, although the number of completed human genomes is still
relatively small, the analysis 
from~\cite{baldicompression07,baldidcc08,brandon2009data}
suggests that, even though DNA
sequences have lengths of over 3 billion base pairs, their compressed
versions should be representable using sets whose cardinalities are
measured in the millions.

\section{Scenarios for Genomic Data}
\label{sec:privacy}
\medskip\noindent
To address security, integrity, and privacy issues for 
querying genomic data,
the general framework we consider comprises two kinds of basic entities:
a data owners, Alice, who controls and stores genomic data, 
and an agent, Bob, who is interested in querying Alice's data. 
The granularity of the data owner can
vary. At the smaller end of the spectrum, each individual could
be a data owner and could, for instance, carry his/her own genome on a
personal storage device. At the larger end of the spectrum, 
Alice could represent a national or
international data center with millions of genomic
sequences. 
Such a data owner, Alice, 
between these ranges could
exist for instance at the state, county, city, hospital, clinic, or
laboratory level. 
Likewise, many different agents, represented as Bob,
are conceivable, ranging from personal physicians, to
hospital, to insurance companies, to employers, to central agencies
such as the Social Security, the FBI, or the CIA.  

\subsection{Query Scenarios}
\medskip\noindent
Let us consider several possible 
query scenarios, where, in each case
Alice holds a list of $n$ entries, each entry containing a genomic
sequence from some person (possibly with $n=1$), and Bob
holds a genome of one individual:
\begin{enumerate}
\item
Bob wants to determine
if there is an entry in Alice's database that matches the genome
he holds, where a match is defined as a symmetric difference that is
below some given threshold, $\tau$.
Bob does not want to reveal his genome to Alice and Alice is willing to reveal
a genome in her database only if its symmetric difference with Bob's
genome is smaller than $\tau$ or close to $\tau$.
\item
Bob wants to determine
if there is an entry in Alice's database that matches the genome
he holds, where a match is defined as a symmetric difference that is
below some given threshold, $\tau$, or close to this threshold.
Bob does not want to reveal his genome to Alice and Alice does not
want to reveal any of her data to Bob.
However,
Alice is willing to reveal a genome in her database
to a trusted third party, Charles, but only
if its symmetric difference with Bob's
genome is smaller than $\tau$.
\item
Bob wants to determine
if there is an entry in Alice's database that matches the genome
he holds, restricted to a range, $R$, 
where a match is defined as a symmetric difference that is
below some given threshold, $\tau$, or close to $\tau$, and is in the range $R$.
Bob does not want to reveal his genome to Alice and Alice is willing to reveal
a genome in her database, restricted to $R$,
only if its symmetric difference with Bob's
genome is smaller than $\tau$ in $R$.
\end{enumerate}

We show in this paper 
how the private invertible Bloom filter, which we discuss next,
can be used in each of these scenarios.

\section{A Privacy-Enhanced IBF}
\medskip\noindent
Our results rely on the construction of a 
\emph{Privacy-enhanced Invertible Bloom Filter} (PIBF),
which extends the Invertible Bloom Filter (IBF) of Eppstein 
and Goodrich~\cite{eg-sirtd-11},
which is itself
a variant of the classic Bloom filter~\cite{bloom1970space} 
and counting Bloom filter~\cite{bmpsv-iccbf-06,fcab-scswa-00}
data structures.
An IBF is a probabilistic way of maintaining a set that 
allows both insertions and deletions. 
Moreover, one of the important features of an IBF is
that it allows the number of inserted
elements to far exceed the capacity of the data structure 
so long as most of the
inserted elements are deleted or subtracted out later. 
Second, an IBF allows the 
elements of the set to be listed back out if the capacity is not too
large. 
In addition, an IBF allows us to represent a large set using a small
IBF, and then to quickly determine the 
elements in the symmetric difference of two sets, given only their
IBF representation, as we now review.

\subsection{A Review of Invertible Bloom Filters}
\medskip\noindent
Let us review the essential components of an Invertible
Bloom Filter (IBF)~\cite{eg-sirtd-11}, as described in~\cite{egs-sos-11}.
We assume that the items being considered belong to a countable
universe, $U$, and that given any element, $x$, we can determine a
unique integer identifier for $x$ in $O(1)$ time (e.g., either by $x$'s
binary representation or through a hash function).
For instance, in the case of a compressed DNA string, $x$
could be a pair, $(i,\delta)$, where $i$ is an index of a reference
string and $\delta$ is a change with respect to that reference.

An IBF consists of a table, $T$, of $t$ cells, for a given
parameter, $t$; a set of
$k$ random hash functions, $h_1,h_2,\ldots,h_k$, which map any 
element $x$ to $k$ \emph{distinct} cells in $T$;
and a random hash function, $g$, which maps any element $x$ to a random
value in the range $[1,2^\lambda]$, where $\lambda$ is a specified
number of bits.

Each cell of $T$ contains the following fields:
\begin{itemize}
\setlength{\itemsep}{0pt}
\item
\textsf{count}: an integer count of the number of items mapped to
this cell
\item
\textsf{idSum}: a sum of all the items mapped to this
cell
\item
\textsf{gSum}: a sum of $g(x)$ values for all items mapped to this
cell.
\end{itemize}
The \textsf{gSum} field is used for checksum purposes.
An IBF supports several simple algorithms 
for item insertion, deletion, and membership queries,
as shown in Figure~\ref{fig:algs}, which is taken
from~\cite{egs-sos-11}.

\begin{figure*}[tp]
\ifFull\small\fi
\noindent
{\sf initialize}$()$:
\begin{algorithmic}[100]
\FOR {$i=0,\ldots,t-1$}
\STATE
$T[i].\mbox{\textsf{count}}\leftarrow 0$
\STATE
$T[i].\mbox{\textsf{idSum}}\leftarrow 0$
\STATE
$T[i].\mbox{\textsf{gSum}}\leftarrow 0$
\ENDFOR
\end{algorithmic}

\medskip
\noindent
{\sf insert}$(x)$:
\begin{algorithmic}[100]
\FOR {each $h_i(x)$ value, for $i=1,\ldots,k$}
\STATE
add $1$ to $T[h_i(x)].\mbox{\textsf{count}}$
\STATE
add $x$ to $T[h_i(x)].\mbox{\textsf{idSum}}$
\STATE
add $g(x)$ to $T[h_i(x)].\mbox{\textsf{gSum}}$
\ENDFOR
\end{algorithmic}

\medskip
\noindent
{\sf delete}$(x)$:
\begin{algorithmic}[100]
\FOR {each $h_i(x)$ value, for $i=1,\ldots,k$}
\STATE
subtract $1$ from $T[h_i(x)].\mbox{\textsf{count}}$
\STATE
subtract $x$ from $T[h_i(x)].\mbox{\textsf{idSum}}$
\STATE
subtract $g(x)$ from $T[h_i(x)].\mbox{\textsf{gSum}}$
\ENDFOR
\end{algorithmic}

\medskip
\noindent
{\sf isMember}$(x)$:
\begin{algorithmic}[100]
\FOR {each $h_i(x)$ value, for $i=1,\ldots,k$}
\IF {$T[h_i(x)].\mbox{\textsf{count}}= 0$ \textbf{and}
      $T[h_i(x)].\mbox{\textsf{idSum}}= 0$  \textbf{and}
      $T[h_i(x)].\mbox{\textsf{gSum}}= 0$}
\STATE 
\textbf{return} \textbf{false}
\ELSIF {$T[h_i(x)].\mbox{\textsf{count}}= 1$ \textbf{and}
      $T[h_i(x)].\mbox{\textsf{idSum}}= x$  \textbf{and}
      $T[h_i(x)].\mbox{\textsf{gSum}}= g(x)$ }
\STATE 
\textbf{return} \textbf{true}
\ENDIF
\ENDFOR
\STATE \textbf{return} ``not determined''
\end{algorithmic}

\medskip
\noindent
{\sf subtract}$(A,B,C)$:
\begin{algorithmic}[100]
\FOR {$i=0$ to $t-1$}
\STATE $T_C[i].\mbox{\textsf{count}} \leftarrow T_A[i].\mbox{\textsf{count}} - T_B[i].\mbox{\textsf{count}}$
\STATE $T_C[i].\mbox{\textsf{idSum}} \leftarrow T_A[i].\mbox{\textsf{idSum}} - T_B[i].\mbox{\textsf{idSum}}$
\STATE $T_C[i].\mbox{\textsf{gSum}} \leftarrow T_A[i].\mbox{\textsf{gSum}} - T_B[i].\mbox{\textsf{gSum}}$
\ENDFOR
\end{algorithmic}

\medskip
\noindent {\sf listItems}$()$:
\begin{algorithmic}[100]
\WHILE {there is an $i\in[1,t]$ such that 
$T[i].\mbox{\textsf{count}} = 1$ \textbf{or} $T[i].\mbox{\textsf{count}} = -1$}
\IF {$T[h_i(x)].\mbox{\textsf{count}}\,=\, 1$ 
     \textbf{and} $T[h_i(x)].\mbox{\textsf{gSum}}\,=\, 
	     g(T[h_i(x)].\mbox{\textsf{idSum}})$ }
\STATE add the item,
$(T[i].\mbox{\textsf{idSum}})$,
to the ``positive'' output list
\STATE call {\sf delete}($T[i].\mbox{\textsf{idSum}}$)
\ELSIF {$T[h_i(x)].\mbox{\textsf{count}}\,=\, -1$ 
       \textbf{and} $-T[h_i(x)].\mbox{\textsf{gSum}}\,=\, 
       g(-T[h_i(x)].\mbox{\textsf{idSum}})$ }
\STATE add the item,
$(-T[i].\mbox{\textsf{idSum}})$,
to the ``negative'' output list
\STATE call {\sf insert}($-T[i].\mbox{\textsf{idSum}}$)
\ENDIF
\ENDWHILE
\end{algorithmic}

\caption{Operations supported by an invertible Bloom filter. All
arithmetic is assume to be in $Z_p$, where $p$ is a prime number chosen
to be larger than any id, $x$, or $g(x)$ value.}
\label{fig:algs}
\end{figure*}

Note that we can take the difference of one IBF, $A$, with a table $T_A$, and
another one, $B$, with table $T_B$, to produce an IBF, $C$, with table $T_C$,
that represents their signed symmetric
difference, with the items in $A\setminus B$ having
positive signs for their cell fields in $C$ and items in $B\setminus A$ having
negative signs for their cell fields in $C$ (we assume that $C$ is initially
empty). 
Moreover,
given an IBF, which may have been produced either through
insertions and deletions or through a subtract operation,
we can list out its contents by repeatedly looking for cells with
counts of $+1$ or $-1$ and removing the items for those cells if they
pass a test for consistency.
This method therefore produces a list of items that had positive
signs and a list of items that had negative signs.
In particular, with respect to
an IBF, $C$, that is the result of a \textsf{subtract}($A,B,C$)
operation, the positive-signed elements belong to $A\setminus B$ and
the negative-signed elements belong to $B\setminus A$.
Eppstein {\it et al.}~\cite{egs-sos-11} establish the following with respect
to an IBF's ability to be used for set differencing.

\begin{theorem}[\cite{egs-sos-11}]
\label{thm:upper}
Suppose $X$ and $Y$ are sets with $m$ elements in their
symmetric difference, i.e., $m=|X\bigtriangleup Y|$, and let
$\epsilon>0$ be an arbitrary real number.
Let $A$ and $B$ be invertible Bloom filters
built from $X$ and $Y$, respectively, such that each IBF has
$\lambda\ge k+\lceil\log k\rceil$ bits in its \textsf{gSum} field, i.e.,
the range of $g$ is $[1,2^\lambda)$,
and each IBF has at least
$2km$ cells, where $k > \lceil \log (m/\epsilon)\rceil+1$ is the
number of hash functions used.
Then the \textsf{listItems} method for the IBF $C$ resulting from the
\textsf{subtract}$(A,B,C)$ method will list all $m$
elements of $X\bigtriangleup Y$ and identify which belong
to $X\setminus Y$ and which belong to $Y\setminus X$ with 
probability at least $1-\epsilon$.
\end{theorem}

\subsection{Making an IBF Privacy-Enhanced}
\medskip\noindent
Our strategies for creating a Privacy-enhanced Invertible Bloom Filter
(PIBF) depend on which scenario we are 
working in. So let us consider various scenarios 
for genomic data comparison.
We describe each of the scenarios below assuming that Alice has a
single compressed DNA sequence. 
In this case when Alice holds several such sequences,
we would repeat the steps for Alice described below for each of her sequences.

\subsubsection{Alice and Bob}
In the first scenario we consider,
Bob wants to determine
if there is an entry in Alice's database that matches the genome
he holds, where a match is defined as a symmetric difference that is
below some given threshold, $\tau$.
Bob does not want to reveal his genome to Alice and Alice is willing to reveal
a genome in her database only if its symmetric difference with Bob's
genome is smaller than $\tau$ or close to $\tau$.

In order to create a Privacy-enhanced IBF (PIBF), for this scenario,
Bob modifies the 
\textsf{initialize} method of the IBF data structure, as shown in
Figure~\ref{fig:algs}.
In particular, rather than initialize each field $f$ in each cell of the
IBF, $T$, with $0$, Bob now initializes each field $f$ in each cell 
of $T$ with a random number, $r_f$ (independently for each field of each
cell).
He then creates a copy of this IBF, $T_0$, and stores it away for
later.
Next, starting from $T_0$,
Bob constructs an IBF, $B$, for the compressed version of
his DNA sequence, by calling the \textsf{insert} method for each of
the items in the compressed representation of his string.
He then sends the resulting IBF, $B$, to Alice.
Once Alice has received $B$, she removes each item in the compressed
version of her string by calling
the \textsf{delete} operation, shown in Figure~\ref{fig:algs}, on $B$ for each 
such element. She then sends the result, $C$, to Bob.
Once Bob has received $C$ from Alice, he subtracts $T_0$ from $C$
using the \textsf{subtract} method shown in Figure~\ref{fig:algs},
to get an unmasked IBF, $D$.
Then he calls
the \textsf{listItems} method on $D$, to decode the elements of the symmetric
difference between Alice and Bob's sequences, if
this symmetric difference is small enough for the IBF, $D$, to be
decoded using the \textsf{listItems} method.

In this scenario, Alice can learn nothing from the IBF, $B$, since
every possible value of every field of every cell is equally likely
in $B$, given its initialization.
That is, much in the same way that a one-time pad encryption is perfectly
secure, in the information theoretic sense, so too is the IBF $B$ that is sent
to Alice also perfectly secure.
Alice then performs a series of operations on $B$ and sends it back
to Bob. The degree to which this result, $C$, carries information
that Bob can decode using the \textsf{listItems} method depends on
how many elements remain in $D$ (after Alice has removed the elements
of her set from $B$).
We quantify this privacy protection provided to Alice
below, in Theorem~\ref{thm:lower}, from our section, \ref{sec:analysis},
on the analysis of the PIBF.

\subsubsection{Alice, Bob, and Charles}
In the next scenario we consider,
Bob wants to determine
if there is an entry in Alice's database that matches the genome
he holds, where a match is defined as a symmetric difference that is
below some given threshold, $\tau$, or close to this threshold.
Bob does not want to reveal his genome to Alice and Alice does not
want to reveal any of her data to Bob.
However, Alice is willing to reveal a genome in her database
to a trusted third party, Charles, but only
if its symmetric difference with Bob's
genome is smaller than $\tau$.

Unlike the previous scenario, in this scenario, we utilize 
the functionality of homomorphic encryption.
Fortunately, we don't need a complex general-purpose form of
homomorphic encryption.
Instead, we simply need a cryptosystem having encryption and
decryption functions, $E$ and $D$, that satisfy the following
conditions:
\begin{eqnarray}
D(E(x)*E(y)) &=& x+y \\
D(-1\cdot E(x)) &=& -x
\end{eqnarray}
for some effectively computable operations, ``$*$,'' and ``$\cdot$.''
We refer to these operations as \emph{homomorphic addition} and 
\emph{homomorphic inversion}.
For example, the homomorphic cryptosystem of Paillier~\cite{pailli99}
supports such conditions.
In addition, we
assume that the cryptosystem can be made semantically secure, so that
it is computationally difficult to distinguish $E(0)$, $E(1)$, and $E(x)$,
for $x>1$.

Unlike the previous scenario, in this scenario we do not 
modify 
the \textsf{initialize} method; that is,
Bob initializes this IBF by setting each field
of each cell to $0$.
He then constructs an IBF, $B$, for the compressed version of
his DNA sequence, by calling the \textsf{insert} method for each of
the items in the set representing the compressed version of his
string.
He then encrypts each cell of $B$ using a public-key 
homomorphic function, $E$, using the public key of Charles,
and he sends the result, $E(B)$, to Alice.
Alice then encrypts each element
of the compressed version of a string in her database (using the
public-key of Charles), and 
performs the \textsf{delete} operation, but with each arithmetic
operation now performed using homomorphic addition
and inversion for each such element. She then sends the result, $E(C)$, 
to Charles.
Charles can decrypt $E(C)$, 
but he will only be able to successfully perform the
\textsf{listItems} method, to decode the elements of the symmetric
difference between Alice and Bob's sequences, if
this symmetric difference is small enough for the IBF, $C$, so that
the \textsf{listItems} method succeeds.
The crucial observation with respect to the essential \textsf{delete}
method for 
an IBF, used in this scenario,
is that we can perform all the arithmetic operations 
of these operations just using
addition and inversion.
The methods, \textsf{isMember} and \textsf{listItems}, also require
comparisons, however. So we cannot perform these operations using
homomorphic encryption.
That is why they 
must be performed by the trusted third party, Charles.

\subsubsection{Query Ranges}
In our final scenario, we make a simple observation, which can be
applied to either of the two scenarios just described.
In this scenario,
Bob wants to determine
if there is an entry in Alice's database that matches the genome
he holds, restricted to a range, $R$, 
where a match is defined as a symmetric difference that is
below some given threshold, $\tau$, or close to $\tau$, and is in the range $R$.
Bob does not want to reveal his genome to Alice and Alice is willing to reveal
a genome in her database, restricted to $R$,
to either Bob or Charles (depending on her privacy preference)
only if its symmetric difference with Bob's
genome is smaller than $\tau$ in $R$.

In this scenario, Bob and Alice perform their operations as in the
appropriate version of one of the
above scenarios, but they only consider their compressed DNA
sequences as restricted to $R$.
Thus, in each scenario, we achieve the privacy restrictions with
respect to Bob and Alice, but now restricted to their symmetric
difference in $R$.
The important observation of why this is possible is that restriction
to the range $R$ is just a subset operation, and all the correctness
and privacy quantification, which we discuss next, applies equally
well to subsets as to sets.

\subsection{Analysis}
\label{sec:analysis}
\medskip\noindent
We have already discussed how Bob can achieve perfect privacy, in the
information theoretic sense, by initializing each field of each cell
in his IBF to a random number.
Likewise, in the scenario involving a trusted third party, Charles,
Bob's privacy is protected to the degree that the homomorphic
encryption scheme he uses is vulnerable to a known ciphertext attack.
Note that in all our scenarios that the communication complexity
between the parties is always $O(\tau)$, where $\tau$ is the
threshold of interest for the cardinality of the symmetric difference
between Bob's compressed DNA sequence and each compressed DNA
sequence that Alice compares it to, assuming that we use a constant
number of hash functions in the IBFs.

To see that we also achieve useful privacy restrictions with respect to
Bob and Charles, and how much of Alice's string they learn in our
various scenarios,
we provide the following theorem, which quantifies
the privacy our schemes achieve in terms of how much information
Alice leaks when she sends a differenced IBF to Bob or Charles.
We note that such a theorem is not included in previous work in
invertible Bloom filters, since these previous works
were focused on quantifying when
IBFs can be successfully decoded, not when they achieve this degree
of privacy protection.

\begin{theorem}
\label{thm:lower}
Suppose that $n$ elements are stored in an IBF with $m$ cells 
and $k\ge 2$ hash functions, 
and let $\epsilon>0$ be an arbitrary parameter. 
Suppose in addition that
\[
n\ge 1+(m/k)(\ln m+\ln\ln m+\ln k+\ln 1/\epsilon).
\]
Then with probability at least $1-\epsilon$, 
the decoding algorithm is unable to decode any of the cells of the IBF.
\end{theorem}

\begin{proof}
Suppose that $n$ is equal to the given bound rather than greater than it.
Consider a single cell $c$ in the IBF. The probability that a particular element occupies $c$ is $k/m$, so the probability that exactly one element occupies it is
\[
n \frac{k}{m} \left(1-\frac{k}{m}\right)^{n-1}
<\frac{nk}{m}\exp\left(-\frac{(n-1)k}{m}\right).
\]
The probability that there exists a cell that is occupied only once is at most $m$ times this quantity,
$$
nk\exp(-(\ln m+\ln\ln m+\ln k+\ln 1/\epsilon)
= \frac{n\epsilon}{m\log m}.$$
For the given value of $n$, and with the assumption that $k\ge 2$, $n\le m\ln m$, so this probability is at most $\epsilon$,
as desired. If $n$ is greater than the given bound rather than being equal to it, the probability of a single-occupancy cell only decreases, so this assumption is without loss of generality.
\end{proof}

A very similar bound was given by 
Erd\H{o}s and R\'enyi~\cite{ErdRen-MTAMKIK-61} for the coupon 
collector's problem in which one must bound the 
number of independent random draws from $m$ elements 
in order to cover each element twice 
(or more generally some fixed number of times).
Their analysis suggests that the number of elements needed to achieve 
a given failure probability should depend only doubly logarithmically 
on $\epsilon$, rather than singly logarithmically as in our analysis. 
However their bound does not directly apply to our problem, because 
the $kn$ cells chosen by the IBF are not independent of each other,
due to the requirement that an element be stored in $k$ distinct cells.

Together, Theorems~\ref{thm:upper} and~\ref{thm:lower}
provide high-assurance quantifications of the privacy enhancement
that can be achieved using our schemes with respect to Alice's data.
In particular, consider the common feature of each 
scenario, where Bob has a query
sequence, $Q$, which is either an entire compressed DNA sequence
or a subsequence specified by a range, $R$, and he wished to find if
it is a good match with a sequence, $Y$, owned by Alice.
Suppose further that he specifies a cardinality, $\tau$, for the
symmetric difference, so that if $|Q\Delta Y|\le \tau$, then he
(or Charles) should learn $Y$.
Finally, suppose he would like a confidence of 99\% for his findings,
so we set $\epsilon=1/100$.
By Theorem~\ref{thm:upper}, Bob should in this case define 
his IBF so that it has $2k\tau$ cells,
where 
\[
k=\lceil\log (\tau/\epsilon)\rceil+1 
= \lceil\log (100\tau)\rceil+1 
\]
is the number of hash
functions used.
For example, if $\tau=100$, then
he should use 15 hash functions and create an IBF with 3000 cells,
which is clearly much less than an array of 3 billion cells he would
need if he were using a scheme whose size depended on the size of the
entire set instead of the cardinality of interest for his set
difference.

In addition, by Theorem~\ref{thm:lower}, Alice can take comfort that
if the number of differences between $Q$
and her set, $Y$, representing her compressed DNA sequence,
is high enough,
then with 99\% assurance, Bob will not be able to decode any elements
from the IBF he receives from her.
In this case, since his IBF has $2k\tau$ cells, then Alice has a 99\%
assurance of privacy so long as her set $Y$ has at least
\[
n\ge 1+ 2\tau(\ln (2k\tau) + \ln\ln (2k\tau) + \ln k + \ln 100)
\]
differences with $Q$.
For example, if $\tau=100$, and Bob
uses 15 hash functions and creates an IBF with 3000 cells,
then Alice will have a 99\% level of confidence Bob is unable to use
the \textsf{listItems} method to learn any element of her set so long
as the number of differences is at least 3461.
This degree of confidence and privacy enhancement is therefore
quite reasonable in our application domain, 
given that DNA sequences have lengths measured in
billions and are likely to have 
compressed cardinalities, with respect to a reference
string, measured in the millions.

\section{Conclusion}
\label{sec:conclusion}
\medskip\noindent
There is a real and growing need for
techniques that can work in conjunction with genomic
compression and storage technologies to answer queries in a way that
simultaneously preserves the integrity and privacy of the data being
queried 
and the proprietary information contained in the queries themselves.
This paper showed how to perform a critical comparison
operation on DNA sequences
in this framework 
using techniques from algorithm design, cryptography, and
bioinformatics, in a scalable way.

A central theme throughout this paper is that effective solutions
that address trust and privacy for genomic data will need to
fully integrate these three areas, so as to
provide protocols that operate
on genomic data in its existing state while also being efficient
enough to work in real-world applications.

There are several areas for future research on this topic.
In particular, once genomic data is stored (in compressed form) in efficient
data structures and one has a means for genomic data
to be queried with high integrity in a privacy-preserving fashion, 
it is natural to then study methods for
performing various useful computations on this data.
For instance, computations which would be of interest include
understanding and predicting protein structures from DNA, 
modeling and understanding metabolic, signaling, and regulatory networks,
and understanding genome evolution.
Such computations clearly go beyond simple genomic comparisons, and
often involve sophisticated machine learning and data mining
algorithms being applied to collections of genomic data.
Thus, if we are to perform these computations in 
a privacy-preserving high-assurance framework, we need new algorithms
that support these computations without sacrificing privacy and integrity.

\subsection*{Acknowledgments}
This research was supported in part by the National Science
Foundation under grants 
0953071 and 1011840.

{\raggedright
\bibliographystyle{abbrv}
\bibliography{%
baldi,%
bloom,%
cgt,%
clone,%
crypto,%
crypto2,%
crypto3,%
data,%
extra,%
genome,%
geom,%
goodrich,%
info,%
k_anonymity,%
math,%
math2,%
paper,%
ref,%
security,%
security2,%
sigproc,%
social}

\begin{thebibliography}{10}

\bibitem{ae-nepps-07}
A.~Amirbekyan and V.~Estivill-Castro.
\newblock A new efficient privacy-preserving scalar product protocol.
\newblock In {\em AusDM '07: 6th Australasian Conf. on Data Mining and
  Analytics}, pages 209--214, 2007.

\bibitem{akd-spsc-03}
M.~J. Atallah, F.~Kerschbaum, and W.~Du.
\newblock Secure and private sequence comparisons.
\newblock In {\em WPES '03: ACM Workshop on Privacy in the Electronic Society},
  pages 39--44, 2003.

\bibitem{al-sosc-05}
M.~J. Atallah and J.~Li.
\newblock Secure outsourcing of sequence comparisons.
\newblock {\em Int. J. Inf. Secur.}, 4(4):277--287, 2005.

\bibitem{adt-ismsh-11}
G.~Ateniese, E.~De~Cristofaro, and G.~Tsudik.
\newblock (if) size matters: Size-hiding private set intersection.
\newblock In D.~Catalano, N.~Fazio, R.~Gennaro, and A.~Nicolosi, editors, {\em
  Public Key Cryptography (PKC)}, volume 6571 of {\em LNCS}, pages 156--173.
  Springer, 2011.

\bibitem{baldicompression07}
P.~Baldi, R.~W. Benz, D.~Hirschberg, and S.~Swamidass.
\newblock Lossless compression of chemical fingerprints using integer entropy
  codes improves storage and retrieval.
\newblock {\em Journal of Chemical Information and Modeling}, 47(6):2098--2109,
  2007.

\bibitem{bloom1970space}
B.~H. Bloom.
\newblock {Space/time trade-offs in hash coding with allowable errors}.
\newblock {\em Commun. ACM}, 13(7):422{--}426, 1970.

\bibitem{bmpsv-iccbf-06}
F.~Bonomi, M.~Mitzenmacher, R.~Panigrahy, S.~Singh, and G.~Varghese.
\newblock {An improved construction for counting Bloom filters}.
\newblock In {\em Proc. 14th Eur. Symp. on Algorithms}, volume 4168 of {\em
  LNCS}, pages 684{--}695. Springer-Verlag, 2006.

\bibitem{brandon2009data}
M.~Brandon, D.~Wallace, and P.~Baldi.
\newblock Data structures and compression algorithms for genomic sequence data.
\newblock {\em Bioinformatics}, 25(14):1731, 2009.

\bibitem{HapMap07}
T.~I.~H. Consortium.
\newblock A second generation human haplotype map of over 3.1 million snps.
\newblock {\em Nature}, 449:851--861, 2007.

\bibitem{cryptoeprint:2011:141}
E.~D. Cristofaro, P.~Gasti, and G.~Tsudik.
\newblock Fast and private computation of set intersection cardinality.
\newblock Cryptology ePrint Archive, Report 2011/141, 2011.
\newblock \url{http://eprint.iacr.org/}.

\bibitem{dmry-erpsi-09}
D.~Dachman-Soled, T.~Malkin, M.~Raykova, and M.~Yung.
\newblock Efficient robust private set intersection.
\newblock In M.~Abdalla, D.~Pointcheval, P.-A. Fouque, and D.~Vergnaud,
  editors, {\em Applied Cryptography and Network Security (ACNS)}, volume 5536
  of {\em LNCS}, pages 125--142. Springer, 2009.

\bibitem{dfknt-uscfm-06}
I.~Damg{\r{a}}rd, M.~Fitzi, E.~Kiltz, J.~B. Nielsen, and T.~Toft.
\newblock Unconditionally secure constant-rounds multi-party computation for
  equality, comparison, bits and exponentiation.
\newblock In S.~Halevi and T.~Rabin, editors, {\em Theory of Cryptography},
  volume 3876 of {\em LNCS}, pages 285--304. Springer, 2006.

\bibitem{dt-ppsip-10}
E.~De~Cristofaro and G.~Tsudik.
\newblock Practical private set intersection protocols with linear complexity.
\newblock In R.~Sion, editor, {\em Financial Cryptography and Data Security},
  volume 6052 of {\em LNCS}, pages 143--159. Springer, 2010.

\bibitem{da-smc-01}
W.~Du and M.~J. Atallah.
\newblock Secure multi-party computation problems and their applications: a
  review and open problems.
\newblock In {\em Workshop on New Security Paradigms (NSPW)}, pages 13--22,
  2001.

\bibitem{eg-sirtd-11}
D.~Eppstein and M.~T. Goodrich.
\newblock {Straggler Identification in Round-Trip Data Streams via Newton's
  Identities and Invertible Bloom Filters}.
\newblock {\em IEEE Trans. on Knowledge and Data Engineering}, 23:297{--}306,
  2011.

\bibitem{egs-sos-11}
D.~Eppstein, M.~T. Goodrich, and J.~A. Simons.
\newblock Stepping out of the stream: Algorithmic applications of
  set-differencing data structures.
\newblock Submitted, 2011.

\bibitem{ErdRen-MTAMKIK-61}
P.~Erd{\H{o}}s and A.~R{\'e}nyi.
\newblock {On a classical problem of probability theory}.
\newblock {\em Magyar Tud. Akad. Mat. Kutat{\'o} Int. K{\"o}zl.}, 6:215{--}220,
  1961.

\bibitem{fcab-scswa-00}
L.~Fan, P.~Cao, J.~Almeida, and A.~Z. Broder.
\newblock {Summary cache: a scalable wide-area web cache sharing protocol}.
\newblock {\em IEEE/ACM Trans. Networking}, 8(3):281{--}293, 2000.

\bibitem{FNP04}
M.~Freedman, K.~Nissim, and B.~Pinkas.
\newblock Efficient private matching and set intersection.
\newblock In {\em Advances in Cryptology --- EUROCRYPT 2004.}, 2004.

\bibitem{gmw-hpamg-87}
O.~Goldreich, S.~Micali, and A.~Wigderson.
\newblock How to play any mental game.
\newblock In {\em STOC '87: Proceedings of the nineteenth annual ACM symposium
  on Theory of computing}, pages 218--229, New York, NY, USA, 1987. ACM.

\bibitem{goldstein05}
D.~Goldstein and G.~Cavalleri.
\newblock Genomics: Understanding human diversity.
\newblock {\em Nature}, 437:1241--1242, 2005.

\bibitem{harihara92}
S.~Harihara, M.~Hirai, Y.~Suutou, K.~Shimizu, and K.~Omoto.
\newblock Frequency of a 9-bp deletion in the mitochondrial {DNA} among {A}sian
  populations.
\newblock {\em Human Biology}, 64(2):161--166, 1992.

\bibitem{hn-esopm-10}
C.~Hazay and K.~Nissim.
\newblock Efficient set operations in the presence of malicious adversaries.
\newblock In P.~Nguyen and D.~Pointcheval, editors, {\em Public Key
  Cryptography (PKC)}, volume 6056 of {\em LNCS}, pages 312--331. Springer,
  2010.

\bibitem{hinds05}
D.~A. Hinds, L.~L. Stuve, G.~Nilsen, E.~Eskin, D.~Ballinger, K.~Frazer, and
  D.~Cox.
\newblock Whole genome patterns of common dna variation in three human
  populations.
\newblock {\em Science}, 307:1072--1079, 2005.

\bibitem{baldidcc08}
D.~S. Hirschberg and P.~Baldi.
\newblock Effective compression of monotone and quasi-monotone sequences of
  integers.
\newblock In {\em Proceedings of the 2008 Data Compression Conference (DCC
  08)}. IEEE Computer Society Press, Los Alamitos, CA, 2008.
\newblock In press.

\bibitem{jks-tppgc-08}
S.~Jha, L.~Kruger, and V.~Shmatikov.
\newblock Towards practical privacy for genomic computation.
\newblock In {\em IEEE Symp. on Security and Privacy}, pages 216--230, 2008.

\bibitem{ks-ppso-05}
L.~Kissner and D.~Song.
\newblock Privacy-preserving set operations.
\newblock In V.~Shoup, editor, {\em Advances in Cryptology (CRYPTO)}, volume
  3621 of {\em LNCS}, pages 241--257. Springer, 2005.

\bibitem{levy07}
S.~Levy, G.~Sutton, P.~C. Ng, and et~al.
\newblock The diploid genome sequence of an individual human.
\newblock {\em PLOS Biology}, 5(10):2113--2144, 2007.

\bibitem{mrk-zks-03}
S.~Micali, M.~Rabin, and J.~Kilian.
\newblock Zero-{K}nowledge sets.
\newblock In {\em Proc.\ 44th IEEE Symposium on Foundations of Computer Science
  (FOCS)}, pages 80--91, 2003.

\bibitem{ors-ecpgqcd-04}
R.~Ostrovsky, C.~Rackoff, and A.~Smith.
\newblock Efficient consistency proofs for generalized queries on a committed
  database.
\newblock In {\em Proc.\ 31st International Colloquium on Automata, Languages
  and Programming (ICALP)}, pages 1041--1053, 2004.

\bibitem{pailli99}
P.~Paillier.
\newblock Public-key cryptosystems based on composite residuosity classes.
\newblock In J.~Stern, editor, {\em Advances in Cryptology --- EUROCRYPT '99},
  volume 1592 of {\em Lecture Notes in Computer Science}, pages 223--239.
  Springer Verlag, 1999.

\bibitem{ss-ppsip-07}
Y.~Sang and H.~Shen.
\newblock Privacy preserving set intersection protocol secure against malicious
  behaviors.
\newblock In {\em 8th Int. Conf. on Parallel and Distributed Computing,
  Applications and Technologies (PDCAT)}, pages 461--468, 2007.

\bibitem{ss-ppsib-08}
Y.~Sang and H.~Shen.
\newblock Privacy preserving set intersection based on bilinear groups.
\newblock In {\em 31st Australasian Conf. on Computer science (ACSC)}, pages
  47--54, 2008.

\bibitem{spol-tpdp-06}
D.~Szajda, M.~Pohl, J.~Owen, and B.~G. Lawson.
\newblock Toward a practical data privacy scheme for a distributed
  implementation of the {Smith-Waterman} genome sequence comparison algorithm.
\newblock In {\em Network and Distributed System Security (NDSS)}, 2006.

\bibitem{thomas98}
M.~Thomas, C.~Cook, K.~Miller, M.~Waring, and E.~Hagelberg.
\newblock Molecular instability in the {COII}-t{RNA}(lys) intergenic region of
  the human mitochondrial genome: multiple origins of the 9-bp deletion and
  heteroplasmy for expanded repeats.
\newblock {\em Philosophical Transactions of the Royal Society of London -
  Series B: Biological Sciences}, 353(1371):955--965, 1998.

\bibitem{tkc-pperd-07}
J.~R. Troncoso-Pastoriza, S.~Katzenbeisser, and M.~Celik.
\newblock Privacy preserving error resilient {DNA} searching through oblivious
  automata.
\newblock In {\em 14th ACM Conf. on Computer and Communications Security
  (CCS)}, pages 519--528, 2007.

\bibitem{vc-ssica-05}
J.~Vaidya and C.~Clifton.
\newblock Secure set intersection cardinality with application to association
  rule mining.
\newblock {\em J. Comput. Secur.}, 13(4):593--622, 2005.

\bibitem{y-psc-82}
A.~C. Yao.
\newblock Protocols for secure computations.
\newblock In {\em 23rd Symp. on Foundations of Computer Science (FOCS)}, pages
  160--164, 1982.

\end{thebibliography}
}

\end{document}